%% file: main.tex
\pdfoutput=1
\documentclass[11pt,twoside]{article}
\usepackage[linesnumbered, ruled]{algorithm2e}
\usepackage{jeffe}
\usepackage[hmargin=1in,vmargin=1in]{geometry}
\usepackage{graphicx}
\usepackage{amsmath,amssymb}
\usepackage[caption=false]{subfig}

\newtheorem{property}{Property}
\newtheorem{observation}{Observation}
\newtheorem{definition}{Definition}
\newtheorem{theorem}{Theorem}
\newtheorem{lemma}{Lemma}
\DeclareMathOperator{\EX}{\mathbb{E}}

\begin{document}

\title{Approximating the Geometric Edit Distance%
\thanks{A preliminary version of this work appeared in the Proceedings of the 30th International
Symposium on Algorithms and Computation~\cite{fl-aged-19}.
Most of the work was done while the second author was a student at the University of Texas at
Dallas.}
}


\author{Kyle Fox%
\thanks{The University of Texas at Dallas, \url{kyle.fox@utdallas.edu}}           %  \\
\and
        Xinyi Li%etc.
\thanks{The University of Utah, \url{xin_yi.li@utah.edu}}           %  \\
}


%\institute{Kyle Fox \at
              %The University of Texas at Dallas, USA \\
              %\email{kyle.fox@utdallas.edu}           %  \\
           %\and
           %Xinyi Li \at
              %The University of Utah, USA\\
              %\email{xin\_yi.li@utah.edu} 
%}


\maketitle

\begin{abstract}
Edit distance is a measurement of similarity between two sequences such as strings, point sequences, or polygonal curves. 
Many matching problems from a variety of areas, such as signal analysis, bioinformatics, etc., need to be solved in a geometric space.
Therefore, the geometric edit distance (GED) has been studied.  
In this paper, we describe the first strictly sublinear approximate near-linear time algorithm for computing the GED of two point sequences in constant dimensional Euclidean space.
Specifically, we present a randomized \(O(n\log^2n)\) time \(O(\sqrt n)\)-approximation algorithm. 
Then, we generalize our result to give a randomized $\alpha$-approximation algorithm for any
$\alpha\in [\sqrt{\log n}, \sqrt{n/\log n}]$, running in time $O(n^2/\alpha^2 \log n)$. 
Both algorithms are Monte Carlo and return approximately optimal solutions with high probability. 
%\keywords{Approximation algorithms\and Computational geometry \and Geometric edit distance\and Ordered sequence measurements}
\end{abstract}
%\begin{acknowledgements}
	%The authors would like to thank Anne Driemel and Benjamin Raichel for helpful
	%discussions.
%\end{acknowledgements}
\input{intro}
\input{alg1}
\input{alg2}
\input{O1alg}


\bibliographystyle{spphys}       % APS-like style for physics
%\bibliographystyle{plainurl}
\bibliography{edit}

% end of file template.tex
\end{document}

%% file: intro.tex
	\section{Introduction}
	\label{sec:introduction}
	
Ordered sequences are frequently studied objects in the context of similarity measurements, because sequence alignment plays a vital role in trajectory comparison and pattern recognition. 
	As a consequence, several metrics have been developed to measure the similarity of two sequences, e.g., Fr\'{e}chet distance, dynamic time warping, and their variations. 
	Geometric edit distance, a natural extension of the string metric to geometric space, is the focus of this paper. 
	This concept is formally introduced by Agarwal \textit{et al.} \cite{DBLP:journals/corr/AgarwalFPY15}; however, a similar idea (extending string edit distance to a geometric space) has been applied in other ways during the past decade. 
	Examples include an $l^p$-type edit distance for biological sequence comparison \cite{Stojmirovic2009GeometricAO}, ERP (Edit distance with Real Penalty)~\cite{CHEN2004}, EDR (Edit Distance on Real sequence) \cite{Chen2005}, TWED (Time Warping Edit Distance) \cite{Marteau2009} and a matching framework from Swaminathan \textit{et al.} \cite{Sankararaman:2013:MMS:2525314.2525360} motivated by computing the similarity of time series and trajectories. 
	See also a survey by Wang \textit{et al.}~\cite{Wang2013}.
	
	\paragraph{\textbf{Problem statement.}} 
	Geometric Edit Distance (GED) is the minimum cost of any matching between two geometric point sequences that respects order along the sequences.
	The cost includes a constant penalty for each unmatched point.
	
	Formally, let $P=<p_1, ..., p_m>$ and $Q=<q_1, ..., q_n>$ be two point sequences in ${\rm I\!R}^d$ for some constant $d$. 
	A monotone \textit{matching} $\mathcal{M}$ is a set of index pairs $\{(i_1, j_1), ..., (i_k,
  j_k)\}$ such that the first elements \(i\) (respectively, second elements \(j\)) are distinct and for any two elements $(i, j)$ and $(i', j')$ in $\mathcal{M}$, $i<i'$ if $j< j'$. 

	We call every unmatched point a \textit{gap point}. 
	Let $\Gamma(\mathcal{M})$ be the set of all gap points. 
	The \textit{cost} of $\mathcal{M}$ is defined as
	\begin{align}
	\label{equ:Geocost}
	\delta(\mathcal{M}) = \sum_{(i,j)\in \mathcal{M}}dist(p_i,q_j) + \rho(\Gamma(\mathcal{M}))
	\end{align}
	where $dist(p, q)$ is the distance between points $p$ and $q$ (i.e. the Euclidean norm), and $\rho(\Gamma(\mathcal{M}))$ is a function of all gap points, which is known as a \textit{gap penalty function}. 
	The use of gap points and the gap penalty function allows us to recognize good matchings even in the presence of outlier points. 
	The distance is sensitive to scaling, so, we may only match point pairs that are sufficiently
  close together.
	In this paper, we use a linear gap function.
	That is to say, $\rho(\Gamma(\mathcal{M}))=|\Gamma(\mathcal{M})|\cdot \ell$, where $\ell$ is a constant parameter called the \textit{gap penalty}.
  Without loss of generality, we may assume \(\ell = 1\) when designing algorithms.
	\begin{definition}
		\label{def:GED}
			We denote the GED between two sequences $P, Q$ as:
		\[
		%\label{equ:GED}
		GED(P,Q)=\min_\mathcal{M} \delta(\mathcal{M})=\min_\mathcal{M}\left(\sum_{(i,j)\in \mathcal{M}}dist(p_i,q_j) + |\Gamma(\mathcal{M})| \right)  
		\]
		where the minimum is taken over all monotone matchings. 
		%Without loss of generality, we assume $\ell=1$ throughout the paper. 
	\end{definition}
	\paragraph{\textbf{Prior work.}} 
	To simplify the presentation of prior work, we assume \(n \geq m\).
	It is trivial to compute $GED(P, Q)$ in $O(mn)$ time by simply changing the cost of substitution in the original string edit distance (Levenstein distance) dynamic programming algorithm \cite{Wagner:1974:SCP:321796.321811}.
  Assuming $k$ is the GED, we can achieve an $O(nk)$ time algorithm
  by restricting our attention to the middle \(k\) diagonals of the dynamic programming table (see
  also Ukkonen~\cite{Ukkonen:1985:AAS:4620.4626}).
	There is a slightly subquadratic $O(n^2/\log n)$ time algorithm \cite{Masek1980AFA} for the string
  edit distance, but it appears unlikely we can apply it directly to the geometric case.
	Accordingly, Gold and Sharir~\cite{Gold:2018:DTW:3266298.3230734} proposed a different algorithm which can compute GED as well as the closely related dynamic time warping (DTW) distance in $O(n^2\log\log\log n/\log\log n)$ time in polyhedral metric spaces. 
	Recent papers have shown conditional lower bounds for several sequence distance measures even with some restrictions.
  Assuming the Strong Exponential Time Hypothesis (SETH)~\cite{ip-ock-01}, there is no $O(n^{2-\delta})$ time algorithm for any constant $\delta > 0$ for Fr\'{e}chet distance \cite{DBLP:journals/corr/Bringmann14}, DTW over a constant size alphabet \cite{7354388} or restricted to one-dimensional curves \cite{DBLP:journals/corr/BringmannK15}, and string edit distance on the binary alphabet \cite{DBLP:journals/corr/BackursI14,DBLP:journals/corr/BringmannK15}.%
	\footnote{The (discrete) Fr\'{e}chet and DTW distances are defined similarly to GED;
	however, they use one-to-many correspondences instead of one-to-one matchings, and they disallow the use of gap points.
	As in GED, DTW aims to minimize the sum of distances between corresponding points, while discrete Fr\'{e}chet distance aims to minimize the maximum distance over corresponding points.}
	The latter of the above results implies the same conditional lower bound for GED, even assuming the sequences consist entirely of \(0,1\)-points in ${\rm I\!R}$.
	
	Due to these limitations and difficulties, many researchers have turned to approximation algorithms for these distances.
	In particular, much work has been done to explore approximate algorithms for Fr\'{e}chet distance, DTW, and string edit distance~\cite{DBLP:journals/jocg/BringmannM16,chan2018improved,kuszmaul2019dynamic,DBLP:journals/corr/AgarwalFPY15,DBLP:journals/corr/abs-1005-4033,chakraborty2018approximating}.
	Bringmann and Mulzer~\cite{DBLP:journals/jocg/BringmannM16} describe an \(\alpha\)-approximation algorithm for the discrete Fr\'{e}chet distance that runs in time \(O(n \log n + n^2 / \alpha)\) for any \(\alpha \in [1, n]\).
  Chan and Rahmati~\cite{chan2018improved} improved this running time to $O(n\log n +
  n^2/\alpha^2)$. 
  Very recently, Kuszmaul~\cite{kuszmaul2019dynamic} provided $O(\alpha)$-approximation algorithms
  with $O((n^2/\alpha) \polylog n)$ running times for edit distance over arbitrary metric spaces and
  DTW over well separated tree metrics. 
  Another $O(n^2/\alpha)$ time algorithm with an $O(\alpha)$ approximation factor for \emph{string}
  edit distance is to run Ukkonen’s~\cite{Ukkonen:1985:AAS:4620.4626} $O(nk)$ time algorithm letting
  $k$ be $n/\alpha$, and unmatch all characters if this algorithm cannot return the optimal
  matching. 
  Similarly, we can obtain a different $O(\alpha)$-approximation algorithm for GED running in
  $O(n^2/\alpha)$ time by making use of the $O(nk)$ time exact algorithm mentioned above.
	There are many other approximation algorithms specialized for the string version of edit distance.
	In particular, an $O(\sqrt{n})$-approximation algorithm with linear running time can be acquired easily from an $O(n+k^2)$ time exact algorithm \cite{Landau1998IncrementalSC}.
  More recent results include algorithms with $(\log n)^{O(1/\varepsilon)}$
  \cite{DBLP:journals/corr/abs-1005-4033} and constant approximation ratios
  \cite{chakraborty2018approximating} with different running time tradeoffs.
  The latest result in this line of work is an~\(O(n^{1 + \varepsilon})\) time constant factor
  approximation algorithm for any \(\varepsilon > 0\)~\cite{an-edntc-20}.
	
	For GED specifically, a simple linear time $O(n)$-approximation algorithm was observed by Agarwal \textit{et al.} \cite{DBLP:journals/corr/AgarwalFPY15}. 
	In the same paper, they also offered a subquadratic time (near-linear time in some scenarios) approximation scheme on several well-behaved families of sequences. 
	Using the properties of these families, they reduced the search space to find the optimal admissible path in the dynamic programming graph \cite{DBLP:journals/corr/AgarwalFPY15}. 
	\paragraph{\textbf{Our results.}}
	Inspired by the above applications and prior work, we commit to finding a faster approach to
  approximating GED between general point sequences while also returning the approximate best matching.
  Here, we give the first near-linear time algorithm to compute GED with a strictly sublinear
  approximation factor. 
	We then generalize our result to achieve a tradeoff between the running time and approximation factor. 
  Both of these algorithms are Monte Carlo algorithms, returning an approximately best matching with
  high probability.%
  \footnote{We say an event occurs with high probability if it occurs with probability at least
  $1-\frac{1}{n^c}$ for some constant $c>0$.}
  To simplify our exposition, we assume the points are located in the plane (i.e., $d=2$), and we
  assume the input sequences are the same length (i.e., $m=n$).
  We can easily extend our results to the unbalanced case, and our analysis implies that outside the
  plane, the running times and approximation ratios increase only by a factor polynomial in $d$.
	\begin{theorem}
		\label{the:alg1}
		Given two point sequences $P$ and $Q$ in ${\rm I\!R}^2$, each with n points, there exists an $O(n\log^2n)$-time randomized algorithm that computes an $O(\sqrt{n})$-approximate monotone matching for geometric edit distance with high probability. 
	\end{theorem}
	
	The intuitive idea behind this algorithm is very simple.
  We check if the GED is less than each of several geometrically increasing values $g$, each of which is less than $O(\sqrt{n})$. 
	For each $g$, we transform the geometric sequences into strings using a randomly shifted grid and
  run the $O(n+k^2)$ time exact algorithm for strings~\cite{Landau1998IncrementalSC}.
	If the GED is less than $g$, then we get an $O(\sqrt{n})$ approximate matching.
	If we never find a matching of cost $O(\sqrt{n})$, we simply leave all points unmatched as this empty matching is an $O(\sqrt{n})$-approximation for GED with high probability. 
	We give the details for this $O(\sqrt{n})$-approximation algorithm in Section \ref{sec:alg1}.
	
	\begin{theorem} 
		\label{the:alg2}
		Given two point sequences $P$ and $Q$ in ${\rm I\!R}^2$, each with n points, there exists an $O(\frac{n^2}{\alpha^2}\log n)$-time randomized algorithm that computes an $O(\alpha)$-approximate monotone matching for geometric edit distance with high probability for any $\alpha \in [\sqrt{\log n}, \sqrt{n/\log n}]$. 
	\end{theorem} 
	
	The second algorithm uses similar techniques to the former, except we can no longer use the string edit distance algorithm as a black box.
	In particular, we cannot achieve our desired time-approximation tradeoff by just directly altering some parameters in our first algorithm.
	We discuss why in Section~\ref{subsec:flaws}.
  To overcome these difficulties, we develop a constant-factor approximation algorithm to compute
  the GED of point sequences obtained by snapping points of the original input sequences to grid
  cell corners.
	Our algorithm for these snapped points is based on the exact algorithm for string edit distance \cite{Landau1998IncrementalSC} but necessarily more complicated to handle geometric distances.
	So, we first introduce the $O(n+k^2)$ time algorithm for strings in Section
  \ref{sec:StringAlgorithm}, and then describe our constant approximation algorithm for points in Section \ref{sec:SGED}.
  We note that a key component of the string algorithm and our extension is a fast method for finding
  maximal length common substrings from a given pair of starting positions in two strings~\(A\)
  and~\(B\).
  A similar procedure was needed in the discrete Fr\'{e}chet distance approximation of Chan and
  Rahmati~\cite{chan2018improved}.
  In Section~\ref{sec:trade-off}, we present the algorithm for Theorem~\ref{the:alg2} using our
  constant approximation algorithm for snapped point sequences as a black box.

%% file: alg1.tex
	\section{$O(\sqrt{n})$-Approximation for GED}
	\label{sec:alg1}
	Recall that the main part of our algorithm is a decision procedure to check if the GED is less than a guess value $g$. 
	There are two steps in this process: 
	\begin{enumerate}
		\item Transform the point sequences into strings. 
		To be specific, we partition nearby points into common groups and  distant points into different groups to simulate the identical characters and different characters in the string version of edit distance. 
  \item Run a modification of the exact string edit distance algorithm of Landau \emph{et
    al.}~\cite{Landau1998IncrementalSC}.
    To better serve us when discussing geometric edit distance, we aim to minimize the number of
    insertions and deletions to turn $S$ into $T$ \emph{only};
    we consider substitution to have infinite cost.
    Details on this modified algorithm appear in Section~\ref{sec:StringAlgorithm}.%
    \footnote{This variant of the string edit distance is really closer to the shortest common
    supersequence length of the strings rather than the traditional Levenshtein distance, but we
    stick with ``edit distance'' for simplicity.}

	\end{enumerate}
	We explain how to transform the point sequences into strings in Section \ref{subsec:randomGrid}, and we analyze the approximation factor and running time in Sections \ref{subsection:time} and \ref{subsec:alg1corr}.

	For convenience, we refer to the string edit distance algorithm as $SED(
	S, T, k)$, where $S$ and $T$ are two strings with equal length. This algorithm will return a matching in $O(n+k^2)$ time if the edit distance is at most $k$.  
	We give an outline of our algorithm as Algorithm~\ref{alg:subliner}.
	Here, \(c\) is a sufficiently large constant, and we use \(\lg\) to denote the logarithm of base \(2\).
	\begin{algorithm}[h]
		\SetAlgoLined
		\LinesNumbered
		\KwIn{Point sequences $P$ and $Q$}
		\KwOut{An approximately optimal matching for GED}
		\BlankLine 
		\eIf{$\sum_{i=1}^{n} dist(p_i, q_i) \leq 1$}{
			return matching $\{(1, 1), ..., (n, n)\}$
		}{
			\For{$i := 0$ to $\lceil\lg\sqrt{n}\rceil$}{
				$g := 2^i$
				\\\For{$j:=1$ to $\lceil c\lg n\rceil$}{
					Transform $P$, $Q$ to strings $S$, $T$ using a randomly shifted grid \label{alg2line:transform}
					\\$out := SED(S, T, 12\sqrt{n} + 2g)$ \label{alg2line:callSED}
					\\\If{$out\neq false$}{
						return out
					}
				}
			}
			return the empty matching
		}
		\caption{$O(\sqrt{n})$-approximation algorithm for GED}
		\label{alg:subliner}
	\end{algorithm}
	
	\subsection{\textbf{Transformation by a random grid}}
	\label{subsec:randomGrid}
	As stated above, the transformation technique should partition nearby points into common groups and distant points into different groups.
	We use a randomly shifted grid to realize this ideal;
  see Har-Peled~\cite{HarPeled2010Chapter1R} for an introduction to randomly shifted grids.
	
	Recall $P$ and $Q$ lie in ${\rm I\!R}^2$. 
	We cover the space with a grid. 	
	Let the side length of each grid cell be $\Delta$, and let $b$ be a vector chosen uniformly at random from $[0, \Delta]^2$. 
	Starting from an arbitrary position, the grid shifts $b_i$ units in each dimension $i$. 
	For a point $p$, let $id_{\Delta, b}(p)$ denote the cell which contains $p$ in this configuration. 
	We consider two points $p_1=(x_1, y_1)$, and $p_2 = (x_2, y_2)$ in this space. 
    %Then, the probability that \(p_1\) and \(p_2\) lie in two different cells is
	%We want to know the probability that \(p_1\) and \(p_2\) lie in two different cells, i.e., $P(id_{\Delta, b}(p_1)\neq id_{\Delta, b}(p_2))$.
	
% In one-dimension, $P(id_{\Delta, b}(x)\neq id_{\Delta, b}(y))=\min\{\frac{|x-y|}{\Delta}, 1\}$. See \cite{HarPeled2010Chapter1R} for the proof. 
%	
%	We can extend this result to higher dimensions easily because two points are apart if and only if the projections of points in any dimension are grouped into different cells. 
%	So,
	
	\begin{lemma}
		\label{lem:RandomHighD}
		We have $P(id_{\Delta, b}(p_1)\neq id_{\Delta, b}(p_2))\leq \min\{\frac{|x_1 - x_2| + |y_1 - y_2|}{\Delta}, 1\}$.
	\end{lemma}
	We use this observation in our algorithm and set $\Delta = \frac{g}{\sqrt{n}}$ as each cell's side length. 
	
	\subsection{\textbf{Time complexity}}
	\label{subsection:time}
	We claim the running time for Algorithm \ref{alg:subliner} is $O(n\log^2n)$. 
	Computing $\sum_{i=1}^{n} dist(p_i, q_i)$ takes $O(n)$ time. 
	In the inner loop, the transformation operation (line \ref{alg2line:transform}) takes $O(n)$ time assuming use of a hash table. 
	The running time for $SED(S,T,12\sqrt{n} + 2g)$ is $O(n)$ for $g = O(\sqrt{n})$. 
	Summing over the outer loop and inner loop, the overall running time for Algorithm \ref{alg:subliner} is 
	\[
	%\label{equ:TimeComlx}
	\sum_{i=1}^{\lceil\lg\sqrt{n}\rceil} \sum_{j = 1} ^{\lceil c\lg n\rceil} O(n) = O(n\log^2n).
	\]
	
	\subsection{\textbf{Approximation ratio}}
	\label{subsec:alg1corr}
	In this section, we show that Algorithm \ref{alg:subliner} returns an $O(\sqrt{n})$-approximate matching with high probability. 
	\paragraph{\textbf{Notation.}}
  For any monotone matching $\mathcal{M}$, we define $C_S(\mathcal{M})$ as the cost of the
  corresponding edit operations for $\mathcal{M}$ in the string case and $C_G(\mathcal{M})$ to be
  $\delta(\mathcal{M})$ as defined in \eqref{equ:Geocost} for the geometric case (as stated, there
  is no substitution operation in our modified string case).
  Let $\mathcal{M}_G^*$ be the optimal matching for geometric edit distance, and let $\mathcal{M}_S^*$
  denote the optimal matching under the string configuration during a given iteration of the inner
  for loop.
	Our goal is to establish the relationship between $C_G(\mathcal{M}_G^*)$ and $C_G(\mathcal{M}_S^*)$. 
	
	%\subsubsection{Obtain $\mathcal{M}_S^*$ in string space.} We have the following theorem: 
	\begin{lemma}
		\label{lem:alg1corr}
    Consider an iteration of the outer for loop, and suppose $GED(P, Q) \leq g$.
    With a probability at least $1-\frac{1}{n^c}$, at least one of the $\lceil c\lg n \rceil$
    iterations of the inner for loop will return a matching $\mathcal{M}_S^*$ where
    $C_S(\mathcal{M}_S^*)\leq 12\sqrt{n} + 2g$.
	\end{lemma}

	\begin{proof}
		Let $\mathcal{M}$ be a monotone matching, and let $UM_{\mathcal{M}}$ be the set of unmatched indices.
	There are four subsets of pairs in $\mathcal{M}$: 
	\begin{itemize}
		\item $OC_\mathcal{M}$: In each pair, both indices' points fall into One cell, and the distance between the two points is less and equal to $\Delta = \frac{g}{\sqrt{n}}$ (Close). 
		\item $OF_\mathcal{M}$: In each pair, both indices' points fall into One cell, and the distance between the two points is larger than $\Delta =\frac{g}{\sqrt{n}}$ (Far). 
		\item $DC_\mathcal{M}$:	In each pair, the indices' points are in Different cells, and the distance between the two points is less and equal to $\Delta =\frac{g}{\sqrt{n}}$ (Close). 
		\item $DF_\mathcal{M}$: In each pair, the indices' points are in Different cells and the distance between the two points is larger than $\Delta =\frac{g}{\sqrt{n}}$ (Far). 
	\end{itemize}
	These sets are disjoint, so
	\begin{align}
		\label{equ:costGEDinGeo}
		C_G(\mathcal{M}_G^*) = &|UM_{\mathcal{M}_G^*}| + \sum_{(i, j)\in OC_{\mathcal{M}_G^*}} dist(p_i, q_j) + \sum_{(i, j)\in OF_{\mathcal{M}_G^*}} dist(p_i, q_j) \nonumber
		\\&\qquad +\sum_{(i, j)\in DC_{\mathcal{M}_G^*}} dist(p_i, q_j) + \sum_{(i, j)\in DF_{\mathcal{M}_G^*}} dist(p_i, q_j).
	\end{align}
	Recall that there is no substitution operation in our version of the string edit distance.
  Therefore, to better understand optimal matchings for string edit distance, we unmatch all the
  pairs in $DC_{\mathcal{M}_G^*}$ and $DF_{\mathcal{M}_G^*}$, forming a new matching
  $\mathcal{M}_G^{*'}$. 
	Points in one cell are regarded as identical characters while those in different cells are different characters. 
	Therefore, 
	\begin{align*}
		%\label{equ:costGEDinStr}
		C_S(\mathcal{M}_G^{*'}) &= |UM_{\mathcal{M}_G^*}| + 0 \cdot (|OC_{\mathcal{M}_G^*}| + |OF_{\mathcal{M}_G^*}|) + 2 \cdot (|DC_{\mathcal{M}_G^*}| + |DF_{\mathcal{M}_G^*}|) %\nonumber
		\\&= |UM_{\mathcal{M}_G^*}| + 2\cdot(|DC_{\mathcal{M}_G^*}| + |DF_{\mathcal{M}_G^*|}).
	\end{align*}
	Observe that there are at most $\frac{g}{g/\sqrt{n}} = \sqrt{n}$ pairs in $DF_{\mathcal{M}_G^{*}}$ if $C_G({\mathcal{M}_G^{*}})\leq g$. 
	Therefore,
	\begin{align}
		C_S(\mathcal{M}_S^*) &\leq C_S(\mathcal{M}_G^{*'}) \nonumber
		\\&=|UM_{\mathcal{M}_G^*}| + 2|DC_{\mathcal{M}_G^*}| + 2|DF_{\mathcal{M}_G^*}| \leq g+2\sqrt{n} + 2|DC_{\mathcal{M}_G^*}|  \label{equ:uboptimalSt}
	\end{align}
	For any two points $p_i, q_j$, let $P_{D}(i, j)$ be the probability that $p_i$ and $q_j$ are assigned into different cells. 
	From Lemma \ref{lem:RandomHighD}, we can infer \( P_{D}(i, j) \leq \frac{2dist(p_i,
		q_j)}{g/\sqrt{n}}\).
	
	Then,
	\begin{align}
		\EX(|DC_{\mathcal{M}_G^*}|)  %&= \sum_{(i, j) \text{ is close pair}} X_{i, j}P_C(i, j) \nonumber
		&\leq\sum_{(i, j) \in DC_{\mathcal{M}_G^*}}P_D(i, j) 
		\leq \sum_{(i, j) \in DC_{\mathcal{M}_G^*}} \frac{2dist(p_i, q_j)}{g/\sqrt{n}} \label{equ:expectedValue}
		\\&\leq 2\sqrt{n}. \nonumber
	\end{align}
	Therefore, 
	\[
	\EX(C_S(\mathcal{M}_S^*)) \leq 6\sqrt{n} + g.
	\]
	By Markov's inequality, 
	\[
	P[C_S(\mathcal{M}_S^*) \geq 12\sqrt{n} + 2g] \leq \frac{1}{2}.% \label{equ:Markov}
	\]
	In other words, $SED(S, T, 12\sqrt{n} + 2g)$ will fail with probability at most $\frac{1}{2}$ if $GED(P, Q)\leq g$. 
	So, if we test $SED(S, D, 12\sqrt{n} + 2g)$ $\lceil c\lg n\rceil$ times, at least one iteration will return a value if $GED(P, Q)\leq g$ with a probability greater than or equal to
	\[
	1 - \prod_{1}^{\lceil c\lg n\rceil} P[C_S(\mathcal{M}_S^*)\geq 12\sqrt{n} + 2g] \geq 1-\prod_{1}^{\lceil c\lg n\rceil} \frac{1}{2} = 1-\frac{1}{n^c}.
	\]
	We conclude the proof of Lemma \ref{lem:alg1corr}. 
	\end{proof}

  According to Lemma~\ref{lem:alg1corr}, with high probability, we obtain a matching a
  $\mathcal{M}_S^*$ such that $C_S(\mathcal{M}_S^*)\leq 12\sqrt{n} + 2g$ if \(GED(P, Q) \leq g\).
	%According to Lemma \ref{lem:alg1corr}, if all test procedures return false, we can say $C_G(\mathcal{M}_G^*)>g$ with high probability; otherwise, we obtain a matching $\mathcal{M}_S^*$ and $C_S(\mathcal{M}_S^*)\leq 12\sqrt{n} + 2g$.
	%\subsubsection{Recover $\mathcal{M}_S^*$ in geometric space.} 
	
	We now consider $C_G(\mathcal{M}_S^*)$. Again,  $UM_{\mathcal{M}}$ is the set of unmatched indices for a matching $\mathcal{M}$. Observe, for all $(i, j)\in \mathcal{M}_S^*$, points $p_i$ and $q_j$ lie in the same grid cell. Therefore, $dist(p_i, q_j)\leq \frac{\sqrt{2}g}{\sqrt{n}}$ if $(i, j)\in \mathcal{M}_S^*$.
	We have: 
	\begin{align}
	%\label{inequ:costAPP}
	C_G(\mathcal{M}_S^*) &= |UM_{\mathcal{M}_S^*}| + \sum_{(i, j)\in \mathcal{M}_S^*} dist(p_i, q_j) \label{equ:costStringinGeo}
	\\&\leq 12\sqrt{n} + 2g + n\cdot(\frac{\sqrt{2}g}{\sqrt{n}}) = 12\sqrt{n} + 2g + \sqrt{2}g\sqrt{n}\nonumber
	\end{align}
	If $GED(P, Q)\leq \sqrt{n}$, then, with high probability, we obtain a matching $\mathcal{M}_S^*$
  by the end of the outer for loop iteration where $g \geq GED(P, Q)\geq \frac{1}{2}g$. The cost of this matching is at most $12\sqrt{n} + 2g + \sqrt{2}g\sqrt{n} = O(\sqrt{n})GED(P, Q)$. The same approximation bound holds if $GED(P, Q)>\sqrt{n}$, whether or not we find a matching during the outer for loop. We conclude the proof of Theorem \ref{the:alg1}. 

%% file: alg2.tex
	\section{$O(\alpha)$-Approximation for GED}
	\label{sec:trade-off}

	%\subsubsection{Changing the cell's side length fails.} 
	We now discuss our $O(\alpha)$-approximation algorithm for any $\alpha\in[\sqrt{\log n}, \sqrt{n/\log n}]$. A natural approach for extending our $O(\sqrt{n})$-approximation is using the same reduction to string edit distance but letting the cell's side length be a variable depending on the approximation factor $\alpha$.
  We argue, however, that this approach does not lead to a good approximation.
	\subsection{\textbf{Flaws in modifying the $O(\sqrt{n})$-approximation to achieve a tradeoff}}
	\label{subsec:flaws}
	%\paragraph*{Technical analysis}
  Suppose we try to modify the \(O(\sqrt{n})\)-approximation algorithm by simply changing the side
  length of cells.
	Let $\Delta_\alpha$ be the cell's side length which depends on the approximation factor $\alpha$.
  We need to obtain a matching \(\mathcal{M}_S^*\) with high probability such that
  $C_G(\mathcal{M}_S^*) \leq g\cdot O(\alpha)$ during any iteration of the outer for loop with
  \(GED(P, Q) \leq g\).
	
	There can be at most $n$ matched pairs in $\mathcal{M}_S^*$. Following
  \eqref{equ:costStringinGeo}, we need $n\cdot\Delta_\alpha \leq g\cdot O(\alpha)$, implying
	\[
	\Delta_\alpha \leq O(\frac{g \alpha}{n}).
	\]
	
	On the other hand, we have $C_S(\mathcal{M}_S^*)\leq C_G(\mathcal{M}_S^*)\leq g\cdot O(\alpha)$ in our analysis. We then derived $2\sqrt{n}$\ in \eqref{equ:uboptimalSt} as $2\frac{g}{\Delta_\alpha}$.
	We now need $2\frac{g}{\Delta_\alpha}\leq g\cdot O(\alpha)$, implying 
	\begin{align*}
	\Delta_\alpha\geq \Omega(\frac{1}{\alpha}).
	\end{align*}
	
	This is fine for $\alpha = \sqrt{n}$ or for large values of \(g\). But for small $\alpha$ and small \(g\), we cannot have both inequalities be true.
  Therefore, we require a different approach to obtain our desired approximation factor-running time
  tradeoff.
	\subsection{\textbf{$O(\alpha)$-approximation algorithm based on grid-snapping}}
	\paragraph{\textbf{Grid-snapping.}}
    Instead of simply grouping points by their different cells as in the $O(\sqrt{n})$-approximation
    algorithm, we snap points to the lower left corners of their respective grid cells.
  %corners of the randomly shifted grid. Specifically, 
	%\begin{enumerate}
		%\item For points in one cell, we move all of them to one position. 
		%\item We now snap all points to their own cell's nearest corners. For those tie points, we choose the left, lower corners. 
	%\end{enumerate} 
	Let $P' = <p'_1, ..., p'_n>$, $Q'=<q'_1, ..., q'_n>$ be the sequences after grid-snapping.
  Let \(\Delta\) be the cell side length of the grid.
  We immediately obtain the following observation:
\begin{observation}
	\label{ob1}
	For any $p_i$ and $q_j$ from \(P\) and \(Q\), respectively, we have $dist(p'_i, q'_j) \leq
  dist(p_i, q_j) + 2\sqrt{2}\Delta$. Moreover, if $p_i$ and $q_j$ are in different cells,
  \(dist(p'_i, q'_j)\geq \Delta\).
\end{observation}
%\begin{observation}
	%\label{ob3}
	%If $p_i, q_j$ are in different cell and $dist(p_i, q_j)> \Delta$, $dist(p_i, q_j) - 2\sqrt{2}\Delta \leq dist(p_i', q'_j) \leq dist_{GS}(p_i, q_j) +2\sqrt{2}\Delta$.
%\end{observation}

% The Observation \ref{ob1} is obvious since we move all points inside the same cell into one position. To get Observation \ref{ob2}, notice that for two points $p_i, q_j$ in different cell with the original distance $dist(p_i, q_j)\leq \Delta$, the their corresponding cells are adjacent to each other even after moving and snapping. So, we just list all the possible distances between endpoints of two adjacent cells. For Observation \ref{ob3}, we know that a point may extend at most $\sqrt{2}\Delta$ distance after moving and snapping, so the distance between two points will be extended at most $2\sqrt{2}\Delta$. 

 We can then obtain our $O(\alpha)$-approximation algorithm by altering the bound in the outer loop and the test procedure of Algorithm \ref{alg:subliner}.
	See Algorithm~\ref{alg:tradeoff}.
  Here, $AGED(P', Q', k)$ attempts to Approximate $GED(P', Q')$ given that \(P'\) and \(Q'\) have
  their points on the corners of the grid cells.
	If $GED(P', Q') \leq k$, then it returns an $O(1)$-approximate matching for \(GED(P', Q')\).
	Otherwise, it either returns an \(O(1)\)-approximate matching or it returns false.
	
	\begin{algorithm}[h]
		\SetAlgoLined
		\LinesNumbered
		\KwIn{Point sequences $P$ and $Q$}
		\KwOut{An approximately optimal matching for GED}
		\BlankLine 
		\eIf{$\sum_{i=1}^{n} dist(p_i, q_{i}) \leq 1$}{
			return matching $\{(1, 1), ..., (n, n)\}$
		}{
			\For{$i := 0$ to $\lceil\lg\frac{n}{\alpha}\rceil$}{
				$g := 2^i$
				\\\For{$j:=1$ to $\lceil c\lg n\rceil$}{
				    %$\cal{G} :=$ a randomly shifted grid with cell side length \(\Delta = \frac{g\alpha}{n}\)
				    Obtain $P'$, $Q'$ by doing grid-snapping to $P$, $Q$ based on a randomly shifted grid
            %$\cal{G}$, 
					\\$out := AGED\left(P', Q', (4\sqrt{2}+6)g\right)$
					\\\If{$out\neq false$}{
						return out 
					}
				}
			}
			Return the empty matching
		}
		\caption{$O(\alpha)$-approximation algorithm}
		\label{alg:tradeoff}
		\end{algorithm}
	
	We describe how to implement \(AGED(P', Q', k)\) in Section \ref{sec:SGED}. The running time of our implementation is $O(n + \frac{k^2}{\Delta})$ where \(\Delta\) is the cell side length of the grid.  
	We do grid snapping in $O(n)$ time.
  For each \(g = 2^i\), we use cells of side length \(\frac{g\alpha}{n}\) and set \(k\) to $\left(4\sqrt{2} + 6\right)g$, so the overall running time of our \(O(\alpha)\)-approximation algorithm is 
	\begin{align*}
	O(n)+\sum_{i=0}^{\lceil \lg\frac{n}{\alpha}\rceil }\sum_{j=1}^{\lceil c\lg n \rceil} O(n+\frac{2^i n}{\alpha})& = \sum_{i=0}^{\lceil \lg\frac{n}{\alpha}\rceil} O(n\log n + \frac{2^in}{\alpha}\log n) 
	%\\&=O(\sum_{i=0}^{\lceil \lg\frac{n}{\alpha}\rceil}n\log n + \frac{n}{\alpha}\log n\sum_{i=0}^{\lg\frac{n}{\alpha}} 2^{i}) 
	=O(n\log^2n+\frac{n^2}{\alpha^2}\log n).
	\end{align*}
	
  Notice that if $\alpha < \sqrt{\log n}$, the running time of our algorithm is $O(n^2)$. Thus, we
  could just run the classic $O(n^2)$ dynamic programming algorithm if we need an approximation
  factor to be less than $\sqrt{\log n}$.
  On the other hand, the $O(n\log^2 n)$ in the running time is asymptotically insignificant if
  $\alpha \leq \sqrt{n/\log n}$.
  As a result, the running time is $O(\frac{n^2}{\alpha ^2}\log n)$ for any $\alpha \in [\sqrt{\log
  n}, \sqrt{n/\log n}]$. 
	
	\subsection{\textbf{Proof of correctness} }
	\label{subsec:alg2corr}
	The analysis for the $O(\alpha)$-approximation algorithm is similar to the first algorithm.
  First, we introduce some additional notation to that used in Section \ref{subsec:alg1corr}.
	
	Let $C_{GS}(\mathcal{M})$ be the cost of any monotone matching $\mathcal{M}$ using distances between
	the grid-snapped points of \(P'\) and \(Q'\).
	Let $\mathcal{M}_{GS}^*$ be the optimal matching for $P'$ and $Q'$, i.e., $C_{GS}(\mathcal{M}_{GS}^*) = GED(P', Q')$.
	Let $\mathcal{M}_{AGS}$ be the matching returned by $AGED(P', Q', (4\sqrt{2}+6)g)$.
We have the following lemma.
	\begin{lemma}
		\label{lem:alg2corr}
		If $GED(P, Q)\leq g$, with a probability at least $1-\frac{1}{n^c}$, at least one of the $\lceil c\lg n \rceil$ iterations will return a matching $\mathcal{M}_{AGS}$.
	\end{lemma}
	\begin{proof}
		
    Recall, we said pairs of points are \emph{close} if their distance is less than or equal to
    \(\Delta\).
		Similar to \eqref{equ:costGEDinGeo}, and using Observation \ref{ob1}, we have
		\begin{align*}
		C_{GS}(\mathcal{M}_G^{*}) &= |UM_{\mathcal{M}_G^*}| + 0 \cdot (|OC_{\mathcal{M}_G^*}| + |OF_{\mathcal{M}_G^*}|) 
		\\&\qquad +\sum_{(i, j)\in DC_{\mathcal{M}_G^*}} dist(p'_i, q'_j) + \sum_{(i, j)\in DF_{\mathcal{M}_G^*}} dist(p'_i, q'_j)
		\\&\leq |UM_{\mathcal{M}_G^*}| + \Delta\cdot |DC_{\mathcal{M}_G^*}| + \sum_{(i, j)\in DF_{\mathcal{M}_G^*}} \left(dist(p_i, q_j)+2\sqrt{2}\Delta\right).
		\\&= |UM_{M^*_{G}}| +  \sum_{(i, j)\in DF_{\mathcal{M}_G^*}} dist(p_i, q_j) +  \Delta |DC_{M^*_{G}}|+2\sqrt{2}\Delta|DF_{M^*_{G}}|
		\end{align*}
		If $C_G(\mathcal{M}_G^*)\leq g$, then
    \[
		C_{GS}(\mathcal{M}_{GS}^*)\leq g +  \Delta|DC_{\mathcal{M}_G^*}|+ 2\sqrt{2}\Delta |DF_{\mathcal{M}_G^*}|.
    \]
		We have the same observation for $DF_{M_{G}^*}$ as before, that is there are at most $\frac{g}{\Delta}$ pairs in $DF_{M_{G}^*}$. Using the same algebra as \eqref{equ:expectedValue}, we have $\EX(|DC_{\mathcal{M}_G^*}|)\leq \frac{2g}{\Delta}$.
		So, 
		\begin{align*}
		\EX(C_{GS}(\mathcal{M}_{GS}^*)) \leq g + \Delta\cdot \frac{g}{\Delta} + 2\sqrt{2}\Delta \frac{2g}{\Delta}= 2\sqrt{2}g+3g.
		\end{align*}
		According to Markov's inequality, we know
		\begin{align*}
		P\left(C_{GS}({\mathcal{M}_{GS}^*})\geq \left(4\sqrt{2}+6\right)g\right) \leq \frac{1}{2}.
		\end{align*}
		In Section \ref{sec:SGED}, we prove that if $C_{GS}(\mathcal{M}_{GS}^{*}) = GED(P', Q') \leq (4\sqrt{2}+6)g$, then $AGED(P',Q', (4\sqrt{2}+6)g)$ will return a constant approximate matching $\mathcal{M}_{AGS}$.
		So, if we test $AGED(P', Q', (4\sqrt{2}+6)g)$ $\lceil c\lg n\rceil$ times (using different grids each time), with a probability at least $1-\frac{1}{n^c}$, at least one $AGED(P', Q', (4\sqrt{2}+6)g)$ will return a matching $\mathcal{M}_{AGS}$. We conclude the proof of Lemma \ref{lem:alg2corr}. 
	\end{proof}
	
	Finally, from Observation~\ref{ob1}, for every pair \((i, j)\) in $\mathcal{M}_{AGS}$, we have $dist(p_i,
	q_j)\leq dist(p'_i, q'_j)+2\sqrt{2}\Delta$.
	We can now return points to their original positions:
	\begin{align*}
	C_G(\mathcal{M}_{AGS}) &= |UM_{\mathcal{M}_{AGS}}| + \sum_{(i,j)\in DC_{\mathcal{M}_{AGS}}} dist(p_i, q_j) + \sum_{(i,j)\in DF_{\mathcal{M}_{AGS}}} dist(p_i, q_j) 
	\\&\qquad+ \sum_{(i,j)\in OC_{\mathcal{M}_{AGS}}} dist(p_i, q_j) + \sum_{(i,j)\in OF_{\mathcal{M}_{AGS}}} dist(p_i, q_j)
	\\&\leq |UM_{\mathcal{M}_{AGS}}| + \sum_{(i,j)\in DC_{\mathcal{M}_{AGS}}} dist(p'_i, q'_j) + \sum_{(i,j)\in DF_{\mathcal{M}_{AGS}}} dist(p'_i, q'_j) + \sum_{(i,j)\in OC_{\mathcal{M}_{AGS}}} dist(p'_i, q'_j) 
	\\&\qquad+ \sum_{(i,j)\in OF_{\mathcal{M}_{AGS}}} dist(p'_i, q'_j) + 2\sqrt{2}\Delta \left(\left|DC_{\mathcal{M}_{AGS}}\right|+\left|DF_{\mathcal{M}_{AGS}}\right|+\left|OC_{\mathcal{M}_{AGS}}\right|+\left|OF_{\mathcal{M}_{AGS}}\right|\right)
	\\&\leq O(1)\cdot  (4\sqrt{2}+6)g+ n\cdot 2\sqrt{2}\Delta.
	\end{align*}
	
	Recall, $\Delta=\frac{g\alpha}{n}$.
	If we obtain a matching \(\mathcal{M}_{AGS}\) during an iteration where \(g \geq 
	C_G(M_G^{*})=GED(P, Q)\geq \frac{1}{2}g\), then $C_G(\mathcal{M}_{AGS})\leq O(g\alpha) = O(\alpha)\cdot GED(P, Q)$.
  Finishing with the same argument as in Theorem \ref{the:alg1}, we conclude our proof of Theorem
  \ref{the:alg2}.

%% file: O1alg.tex
	\section{Constant Approximation Algorithm $AGED(P', Q', k)$}
  Recall that our constant factor approximation algorithm for GED of grid corner points is based on
  a known $O(n+k^2)$ time exact algorithm for string edit distance~\cite{Landau1998IncrementalSC}.
  We first describe this exact algorithm for strings, which
  we refer as $SED(S, T, k)$, in Section \ref{sec:StringAlgorithm}. Then in Section \ref{sec:SGED},
  we modify this string algorithm to obtain an $O(1)$-approximate matching for edit distance
  between point sequences $P'$ and $Q'$ assuming the points lie on the corners of grid cells and
  $GED(P', Q')\leq k$.
	\subsection{\textbf{The exact $O(n+k^2)$ string edit distance algorithm}}
	\label{sec:StringAlgorithm}
	
	\paragraph{\textbf{Dynamic programming matrix and its properties.}} 
	\label{subsec:DPmatrix}
	
	Let $S=<s_1, s_2, ... s_n>$ and $T=<t_1, t_2, ..., t_n>$ be two strings of length $n$. 
	Let $D$ be the $(n+1)\times (n+1)$ matrix where $D(i, j )$ is the edit distance between substrings $S_i=<s_1, s_2, .., s_i>$ and $T_j=<t_1, t_2, ..., t_j>$.
	We give a label $h$ to every diagonal in this matrix such that for any entry $(i, j)$ in this diagonal, $j = i+h$. 
	See Fig. \ref{Fig1} (a). 
	\begin{figure}[h]
		\centering
		\includegraphics[scale=0.5]{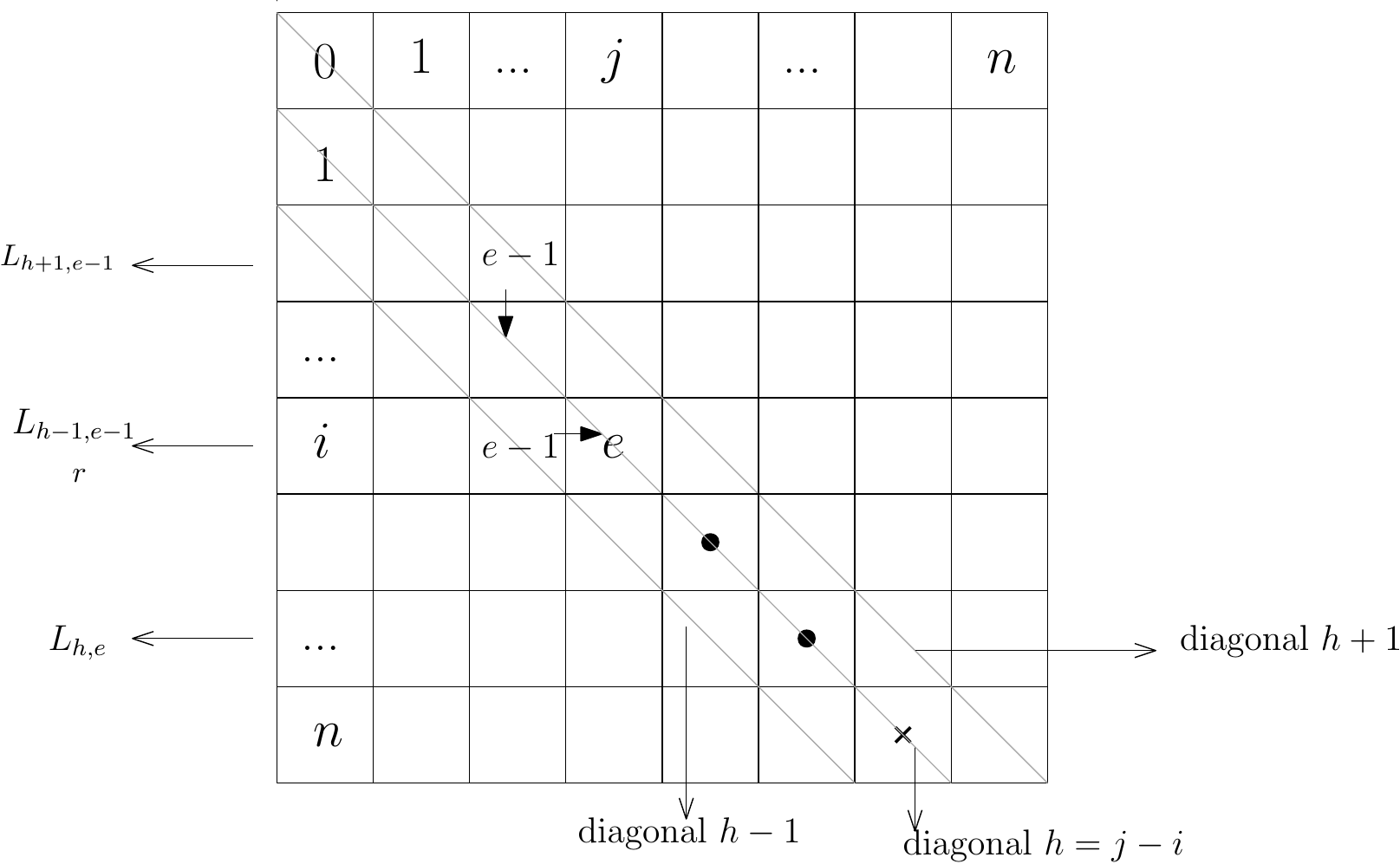}
		\caption{(a) The diagonal containing  any entry $(i, i+h)$ is diagonal $h$
			(b)
			The algorithm slides down the diagonal until finding an entry representing distinct characters.
			A circle means the corresponding two characters are the same; a cross means they are different}
		\label{Fig1}
	\end{figure}
	
	Recall, we aim to minimize only the number of insertions and deletions to turn $S$ into $T$.
  %\emph{only}.
  %; we consider substitution to have infinite cost.\footnote{Computing this variant of the string edit distance is morally closer to computing the longest common subsequnce length of the strings than to traditional Levenshtein distance, but we stick with ``edit distance'' for simplicity.}
	There are four important properties in this matrix which are used in the $O(n+k^2)$ time algorithm. 
	
	\begin{property}
		\label{prop:recursive}
		$
		D(i, j) =\min\begin{cases}
		D(i-1, j) +1  \\
		D(i, j-1) +1 \\
		D(i-1, j-1) + |s_it_j|  
		\end{cases}
		$
		where $|s_it_j| = \begin{cases}
		0, & \text{if } s_i = t_j\\
		\infty, & \text{otherwise}
		\end{cases}.$
	\end{property}
	
	\begin{property}
		\label{prop:init}
		$D(i, 0)  = i$, and $D(0, j)  = j$.
	\end{property}
	
	\begin{property} 
		\label{prop:oddeven}
		$D(i, i+h)$ is even if and only if $h$ is even. 
	\end{property}
	
	\begin{property}
		\label{prop:montone}
		$D(i, j) - D(i-1, j-1) \in \{0, 2\}$.
	\end{property}
	
	Property \ref{prop:montone} can be easily derived from Property \ref{prop:oddeven} and induction on i + j (see Lemma 3 of \cite{Ukkonen:1985:AAS:4620.4626}).
	From Property \ref{prop:montone}, we know all the diagonals are non-decreasing. 
	In particular, all values on diagonal $h$ are greater than $|h|$ considering Property \ref{prop:init}. 
	So, we can just search the band from diagonal $-k$ to $k$ if the edit distance between $S$ and $T$ is at most $k$. 
	
	\paragraph{\textbf{Algorithm for edit distance at most $k$.}}
	\label{subsec:stringAlg}
	We use a greedy approach to fill the entries along each diagonal. 
	For each value $e\in \{0, \dots, k\}$ (the outer loop), we locate the elements whose value is $e$ by inspecting diagonals $-e$ to $e$ (the inner loop). 
	Finally, we return the best matching if $D(n, n)$ is covered by the above search. 
	Otherwise, the edit distance is greater than $k$.
	
	The key insight is that we can implicitly find all entries containing $e$ efficiently in each round. 
	We first define $L_{h, e}$ as the row index of the \textit{farthest} $e$ entry in diagonal $h$. 
	\begin{definition}
		\label{def:maxrow}
		$L_{h, e} = max\{i |D(i, i+h) = e\}$.
	\end{definition}
	
	Note by Property~\ref{prop:oddeven}, \(L_{h,e}\) is well-defined only if \(h \equiv e \mod 2\).
  Observe that all values on diagonal $h$ are at least $|h|$, so we can define our initial values
  as:
	\begin{align*}
		%\label{equ:init}
		L_{h, |h|-2} =
		\begin{cases}
			|h|-1, & \text{if }h< 0;\\
			-1, & \text{otherwise} 
		\end{cases}
		\text{, where } h\in [-k, k].
	\end{align*}
	Let $r = max\{L_{h-1, e-1}, L_{h+1, e-1}+1\}$.
	Then, $D(r, r+h) = e$ by Properties \ref{prop:recursive} and \ref{prop:montone}. 
	Also, if $D(r, r + h) = e$ and $s_{r+1} = t_{r+1+h}$, then $D(r + 1, r+1+h) = e$.
	From these observations, we can compute $L_{h,e}$ in each inner loop using Algorithm
	\ref{alg:computeLde} below.
	
	\begin{algorithm}[h!]
		\SetAlgoLined
		\LinesNumbered
		\BlankLine 
		$r := max\{L_{h-1, e-1}, L_{h+1, e-1}+1\}$
		\\\While{$r+1\leq n$, $r+h+1\leq n$, and $s_{r+1} == t_{r+1+h}$}
		{ \label{alg1line:whilebegin}
			$r := r + 1$ \tcc*{slide} \label{alg1line:slide}
		} \label{alg1line:whileend}
		\eIf{$r> n \text{ or } r+h > n$}{
			$L_{h,e} := \infty$
		}{
			$L_{h,e} := r$
		}
		\caption{Computing $L_{h,e}$ in each inner loop}
		\label{alg:computeLde}
	\end{algorithm}
	
	We call lines \ref{alg1line:whilebegin} through \ref{alg1line:whileend} ``the slide''. 	It is straightforward to recover the optimal matching by using the \(L_{h,e}\) values to trace backwards through the dynamic programming matrix.
	Fig. \ref{Fig1} (b) demonstrates this process.
	%		This algorithm computes $L_{h,e}$ correctly, because:
	%	\begin{enumerate}
	%		\item $L_{h, e} \geq r$, where $r = max\{L_{h-1, e-1}, L_{h+1, e-1} + 1\}$, because the farthest entry of value $e$ must reach at least row $r$ on diagonal $h$ by Properties \ref{prop:recursive} and \ref{prop:montone}.
	%		\item $D(r, r+h) = D(r+1, r+1+h) = e$ only if $s_{r+1} = t_{r+1+h}$.
	%	\end{enumerate}
	
	We can perform slides in constant time each after some $O(n)$-time preprocessing at the beginning of the algorithm.
	Specifically, the length of a slide can be computed using a single lowest common ancestor query in the suffix tree of a string based on \(S\) and \(T\)~\cite{Landau1998IncrementalSC}.
	The overall running time is $O(n+k^2)$. 
	
	%	\subsection{Extension to an $O(\sqrt{n})$-approximation algorithm}
	%	\label{sub:stringApprox}
	%	This algorithm can easily evolve into a linear time $O(\sqrt{n})$-approxima tion algorithm for strings. 
	%	We first check if $SED(S, T, \sqrt{n})$ can return a valid value. 
	%	If it can, we get an exact answer; if it can not, that means the edit distance between the two strings is greater than $\sqrt{n}$. 
	%	In this situation, we just delete all characters in the first string and then insert into this string the same letters as the second sequence. 
	%	The cost of these operations is $2n$, which is an $O(\sqrt{n})$-approximate solution. 
	
	\subsection{\textbf{$O(1)$-approximation algorithm for GED of grid-snapped points}}
	\label{sec:SGED}
	\paragraph{\textbf{Notations.}}
   Similar to the string algorithm, we have a dynamic programming matrix; $D'(i,j)$ is the edit
   distance between subsequences $P'_i = <p'_1, ..., p'_i>$ and $Q'_j = <q'_1, ..., q'_j>$ of points
   at the corners of grid cells. 
	This matrix also meets Property \ref{prop:recursive} stated earlier except that we use $dist(p'_i, q'_j)$ instead of $|s_it_j|$. 
	In addition, we also have the following property which is a refinement of Property~\ref{prop:montone}.
	\begin{property}
		\label{prop:SGEDnon}
		$D'(i, j) - D'(i-1, j-1) \in [0, 2]$.
	\end{property}
	Clearly, the upper bound is 2 (just unmatch $p_i$ and $q_j$). The lower bound can be proved by induction. 
	Because the values in any diagonal are non-decreasing, we need only consider diagonals $-k$ through $k$.
	\paragraph{\textbf{(Implicit) label rules.}}
	To obtain an approximate matching for the edit distance of  snapped point sequences, we now label each entry in the dynamic programming matrix with an approximately tight lower bound on its value.
  Inspired by the string algorithm, we use non-negative integers for our labels, and the entries of
  any diagonal \(h\) only receive labels \(e\) where \(e \equiv h \mod 2\).
	Let $LA(i, j)$ be the label of entry $(i, j)$ and $L'_{h, e}$ be the row index of the farthest entry whose label is $e$ in diagonal $h$. 
	\begin{definition}
		$L'_{h, e} := \max\{i| LA(i, i +h) = e\}$.
	\end{definition}
	%To begin with, we compute the cells whose value is less than 3 in constant time among diagonal $-3$ to $3$.
	%\textcolor{red}{\text{should we demonstrate it in detail?}}
	%If $D(n, n)$ can be discovered in this computing, we can get an exact $SGED(P, Q)$. 
	%\\Otherwise, that means $SGED(P, Q)>3$, 
	For each $e$ from $0$ to $k$, for each diagonal $h$ where \(h \equiv e \mod 2\), we (implicitly)
  assign label $e$ to some entries on diagonal $h$.
	%The labels are approximately exact lower bound on the values of the entries. We use following rules: 
	\begin{enumerate}
		\item  If $h=-e$ or $e$, i.e., this is the first iteration to assign labels to this diagonal,
      then we label the very beginning entry in diagonal $h$ as $e$, i.e., if $h= -e$, $LA(-h, 0) = e$; otherwise, $LA(0, h) = e$. \label{step:init}
		\item We define a \textit{start entry} $(r, r+h)$ for each diagonal $h$. If $h=-e$ or $e$, $r$ is the row index of the first one entry in diagonal $h$ ($r=|h|$ or $0$); otherwise, $r = \max\{L'_{h-1, e-1}, L'_{h+1, e-1}+1\}$. \label{step:startCell}
    \item We assign the label $e$ to entries $(r, r+h)$ to $(r+s, r+h+s)$ where $\sum_{i =r+1}^{s}
      dist(p'_i, q'_{i+h})\leq 2$ and $\sum_{i =r+1}^{s+1} dist(p'_{i}, q'_{i+h})>2$. $L'_{h, e} :=
      r+s$. These entries correspond to a slide in the string algorithm. \label{step:slideLabel}
		\item Finally, if  $(r-1, r+h-1)$ is unlabeled, we go backward up the diagonal labeling entries as $e$ until we meet an entry that has been assigned a label previously.
		(Again, this step is implicit. As explained below, the actual algorithm only finds the \(L'_{h, e}\) entries.)
	\end{enumerate}
	Fig. \ref{fig:GeoDyMatr} illustrates our rules.
	\begin{figure}
		\centering
		\subfloat[Notations and labels for the boundary entries]{\includegraphics[scale=0.3,width=.5\linewidth]{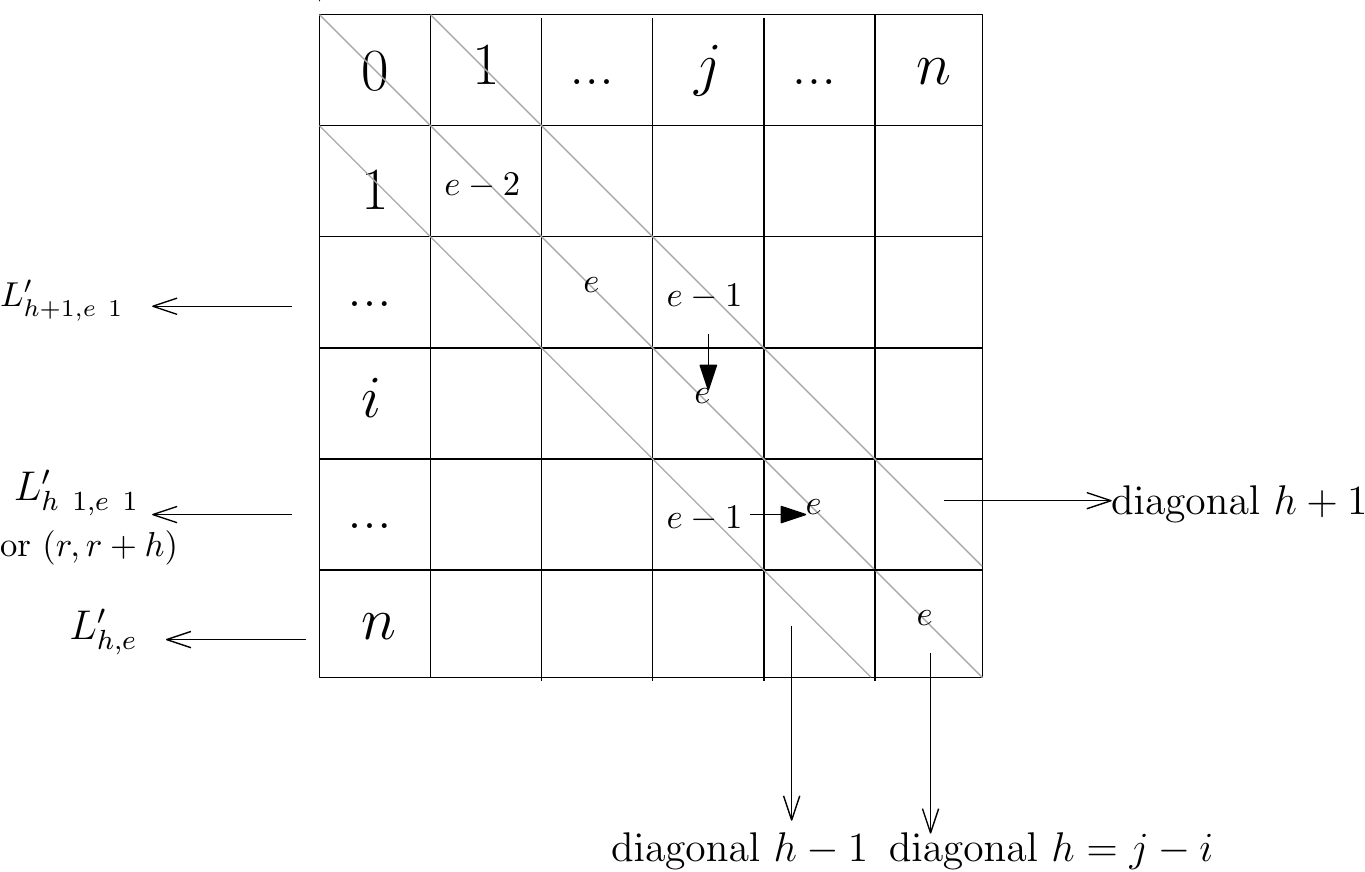}}
		~
		\subfloat[Label entries following step \ref{step:slideLabel}]{\includegraphics[scale=0.23,width= .5\linewidth]{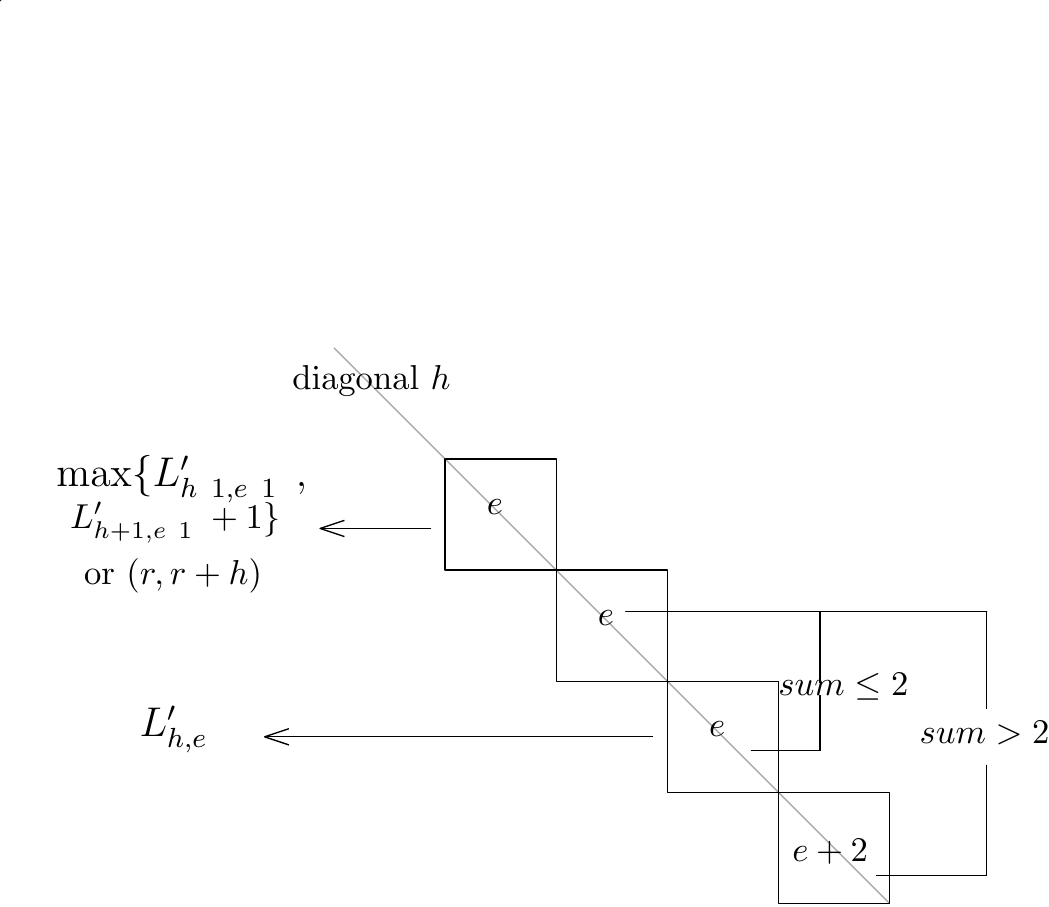}}
		\caption{Notations and rules for approximating SGED}
		\label{fig:GeoDyMatr}
	\end{figure}

\paragraph{\textbf{Computing an approximately optimal matching.}}
	Assume we have set the initial values. Our algorithm only needs to compute each $L'_{h, e}$ as before. See Algorithm \ref{alg:AppSGED}. Then, we guarantee the following theorem:
	\begin{theorem}
		\label{the: O1}
    Suppose \(GED(P', Q') \leq k\).
		We can recover a matching $M_{GS}^{*'}$ using the values $L'_{h, e}$ from Algorithm \ref{alg:AppSGED}.
    The cost of $M_{GS}^{*'}$ for point sequences $P'$, $Q'$ is less and equal to $3GED(P', Q')$.
	\end{theorem}
	In short, we argue each label \(LA(i, j) \leq D'(i, j)\).
	We then follow a path through the matrix as suggested by the way we pick labels in Algorithm \ref{alg:AppSGED}.
	The final matching has cost at most $3LA(n, n)$ which is less and equal to $3GED(P', Q')$.
  The full proof appears in Section \ref{subsec: O1corr}

	\begin{algorithm}[h]
		\SetAlgoLined
		\LinesNumbered
		$r := \max\{(L'_{h-1, e-1}), (L'_{h+1,e-1} + 1)\}$  \label{alg4:unmatch}
				%            \eIf{$s==L'_{h-1, e-1}$}{
				%                $\pi((s, s+h) :=  (L'_{h-1, e-1}, L'_{h-1, e-1}+h-1)$ 
				%            }{
				%            $\pi((s, s+h) :=  (L'_{h+1, e-1}, L'_{h-1, e-1}+h+1)$
				%        }
				%\\$r := r +1$
				%\\$sum:=dist_{SG} (p_{r}, q_{r+h})$
				\\$sum:=0$
				\\\While{$r + 1 \leq n$, $r+h + 1\leq n$, and $\left(sum + dist (p'_{r+1}, q'_{r+h+1}) \leq 2\right)$}{            
					%          $\pi(s, s+h) := (s-1, s+h-1)$ 
					$r:= r +1$ \label{alg4:match}
					\\$sum := sum + dist (p'_{r}, q'_{r+h})$
				}
				\eIf{$r > n$ or $r+h> n$}{
					$L'_{h, e} := \infty$
				}{
					$L'_{h, e} := r$
				}
		%\If{$L'_{0, k} < n$}{return false}
		%{
		%    We can recover a matching $\mathcal{M}_{SG}^{*'}$ by recursively running the $\pi()$ function on cell $(n, n)$. 
		%}
		\caption{Computing $L'_{h, e}$ for the fixed $h$ and $e$}
		\label{alg:AppSGED}
	\end{algorithm}

	We conclude by discussing the time complexity for our algorithm.
	Using the same $O(n)$ preprocessing as in \cite{Landau1998IncrementalSC}, we can slide down maximal sequences of consecutive entries $(r, r+h)$ with $dist(p'_r, q'_{r+h}) = 0$ in constant time per slide.
	Let \(\Delta\) be the cell side length of the grid whose cell corners contain points of \(P'\) and
  \(Q'\).
For $dist(p'_{r}, q'_{r+h}) \neq 0$, we know $dist(p'_{r}, q'_{r+h})\geq \Delta$ from Observations
\ref{ob1}. 
Therefore, we only need to manually add distances and restart faster portions of each slide of distances summing to \(2\) a total of $\frac{2}{\Delta}$ times. 
Thus, the total running time is 
\[
O(n + \sum_{e= 0}^{k}\sum_{h=-e}^{e} \frac{1}{\Delta}) = O(n+\frac{k^2}{\Delta}).
\]
\subsection{\textbf{Proof of Theorem \ref{the: O1}}}
\label{subsec: O1corr}
We have the following properties for our labels along with the following lemma.
\begin{property}
	\label{prop:LAdifference}
	$LA(i, i+h) - LA(i+1, i+1+h) \in \{0, 2\}$.
\end{property}
\begin{property}
	\label{prop:neighbor}
	$LA(i, i+h) - LA(i-1, i+h) \in \{-1, 1\}$ and $LA(i, i+h) - LA(i, i+h-1)\in \{-1, 1\}$.
\end{property}
\begin{lemma}
	\label{lem:constant}
	For every entry $(i, j)$ assigned a label, $LA(i,j)\leq D'(i,j)$.
\end{lemma}
Note that in particular, \(LA(n, n)\leq GED(P', Q')\).
\begin{proof}
	From Property \ref{prop:SGEDnon}, we only need to prove $e$ is the lower bound of the first entry whose label is $e$ in each diagonal $h$. 
	
	We proceed by induction on $e$. 
	\begin{enumerate}
		\item If $e = 0$, we only label the first entry in diagonal 0 as 0. We have $0\leq D'(0 ,0) = 0$. If $e=1$, then for diagonals $1$ and $-1$, we have $1\leq D'(0, 1) = D'(1, 0) =1$.
		%, and for diagonal 0, the value of the  first entry whose label is 1 is less than 1 according to step \ref{step:slideLabel}. 
		\item Assume Lemma \ref{lem:constant} for labels less than $e$. For $e$, we consider the diagonals $h = -e$ to $e$:
		If $h=-e$ or $e$, we know $e\leq D'(|h|, 0)=e$ or $e\leq D'(0, h)=e$. 
		\\Otherwise, let $(f,f+h)$ be the first entry in diagonal \(h\) whose label is $e$. From Property \ref{prop:LAdifference}, $f = L'_{h,e-2}+1$. 
		Fig. \ref{fig:lowerbound} shows the notations. 
		\begin{figure}[h]
			\centering
			\includegraphics[scale=0.4]{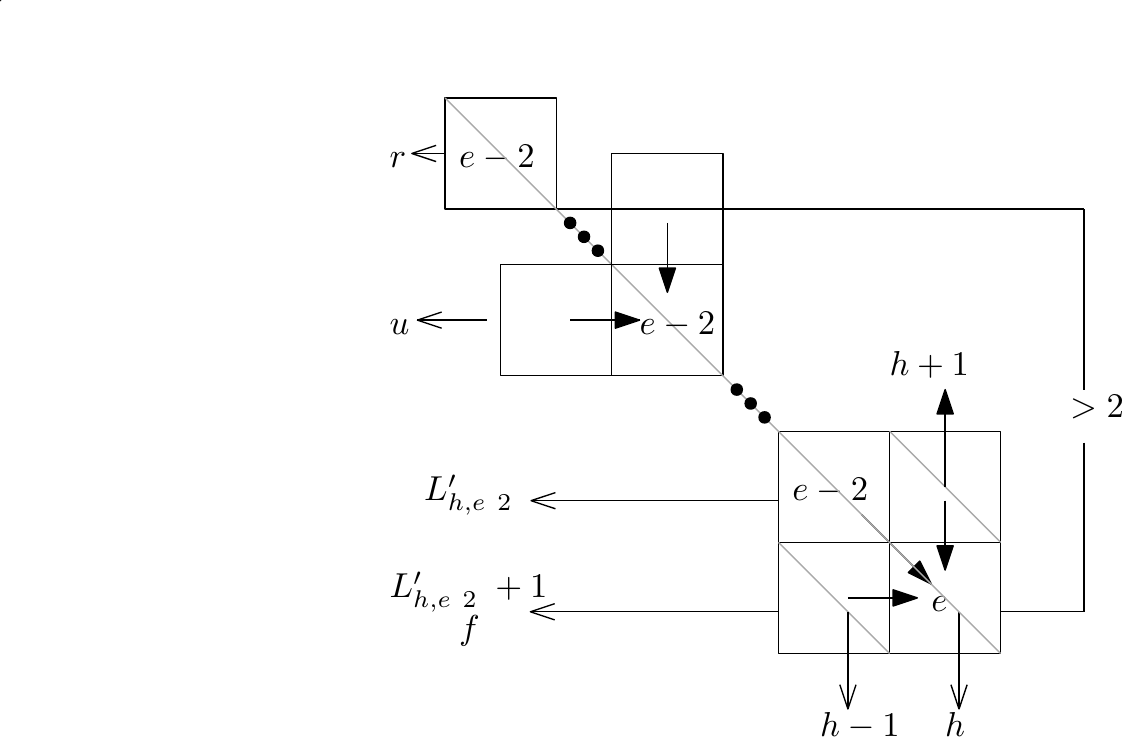}
			\caption{We compute the lower bound of entries which are labeled as $e$}
			\label{fig:lowerbound}
		\end{figure}
		From the refined Property \ref{prop:recursive}, we need to discuss three cases:
		\begin{enumerate}
			\item $D'(f, f+h) = D'(f-1, f+h) + 1$.
			\\From Property \ref{prop:neighbor}, we know $LA(f-1, f+h)=e-1$ or $e+1$.
			\begin{itemize}
				\item[$\bullet$] If $LA(f-1, f+h)=e-1$, $D'(f-1, f+h)\geq e-1$ from our assumption. So, $D'(f, f+h) =  D'(f-1, f+h) + 1 \geq e-1+1 =e$. 
				\item[$\bullet$] If $LA(f-1, f+h)=e+1$, then we know $L'_{h+1, e-1}$ is less than $f-1$. From our assumption and non-decreasing property, $e-1\leq D'(L'_{h+1, e-1}, L'_{h+1, e-1}+h+1)\leq D'(f-1, f+h-1)$. So, $D'(f, f+h)\geq D'(f-1, f+h-1)+1 \geq e$.
			\end{itemize}
			\item $D'(f, f+h) = D'(f, f+h-1) + 1$.
			\\This case is symmetric to the one above. 
			\item $D'(f, f+h) = D'(f-1, f+h-1) + dist(p'_{f}, q'_{f+h})$.
			\\$LA(f-1, f+h-1) = e-2$, because $f-1 = L'_{h, e-2}$. 
			Let $r$ be the row index of the first entry on the slide with label $e-2$ in diagonal $h$, i.e., $r = \max\{L'_{h-1, e-3}, L'_{h+1, e-3}+1\}$. See Fig. \ref{fig:lowerbound}. 
			We define $u$ as the row index of the first entry walking backward from entry $(f, f+h)$ along the diagonal $h$ where $D'(u, u+h) = \min\{D'(u, u+h-1), D'(u-1, u+h)\}+1$.
			\begin{itemize}
				\item[$\bullet$] If $u> r$, like Fig. \ref{fig:lowerbound}, then $u> L'_{h-1, e-3}$ and $u-1> L'_{h+1, e-3}$. 
				We have 
        \[
				D'(u, u+h-1)\geq D'(L'_{h-1, e-3}+1, L'_{h-1, e-3}+h)\geq e-1
				\]
				and
        \[
				D'(u-1, u+h)\geq D'(L'_{h+1, e-3}+1, L'_{h+1, e-3}+h+2)\geq e-1.
				\]
				Therefore,
        \[
				D'(u, u+h) = \min\{D'(u, u+h-1), D'(u-1, u+h)\}+1 \geq e.
        \]
				%implying \(D'(u, u+h)\geq e\).
				%Recall $f\geq u$,
        By Property~\ref{prop:SGEDnon}, we have $D'(f, f+h)\geq e$ as well.
				\item[$\bullet$] If $u\leq r$, then 
				\begin{align*}
				D'(f, f+h) &= D'(r, r+h)+\sum_{i=r+1}^{f} dist(p'_i, q'_{i+h})
				\\&> e-2+2 =e.
				\end{align*}
			\end{itemize}
		\end{enumerate}
		Having considered all the cases, we conclude the proof of Lemma \ref{lem:constant}.
	\end{enumerate}

\end{proof}
\paragraph{\textbf{The bounds for the approximate matching $C_{GS}(\mathcal{M}_{AGS})$.}}
From Algorithm \ref{alg:AppSGED}, we note the label increases correspond to not matching a point in Line \ref{alg4:unmatch}, and slides correspond to matching points. 
Let \(\mathcal{M}_{AGS}\) be the resulting matching.
So, 
\begin{align*}
C_{GS}(\mathcal{M}_{AGS}) &= |UM_{\mathcal{M}_{AGS}}| +\sum_{(i, j)\in \mathcal{M}_{AGS}} dist(p'_i, q'_j) 
\\&\leq LA(n,n) + 2\cdot LA(n,n) \leq 3LA(n,n)\leq 3GED(P', Q').
\end{align*}
We conclude the proof of Theorem \ref{the: O1}.
%and obtain an $O(1)$-approximation algorithm for $GED(P', Q')$. 
%Again, details on the approximation ratio analysis for this algorithm appear in Appendix \ref{app:sec5}.

\paragraph*{Acknowledgements}
	The authors would like to thank Anne Driemel and Benjamin Raichel for helpful
	discussions.

%% file: main.bbl
\begin{thebibliography}{10}
\providecommand{\url}[1]{{#1}}
\providecommand{\urlprefix}{URL }
\expandafter\ifx\csname urlstyle\endcsname\relax
  \providecommand{\doi}[1]{DOI \discretionary{}{}{}#1}\else
  \providecommand{\doi}{DOI \discretionary{}{}{}\begingroup
  \urlstyle{rm}\Url}\fi

\bibitem{fl-aged-19}
K.~Fox, X.~Li, in \emph{Proceedings of the 30th International Symposium on
  Algorithms and Computation} (2019), pp. 26:1--26:16

\bibitem{DBLP:journals/corr/AgarwalFPY15}
P.K. Agarwal, K.~Fox, J.~Pan, R.~Ying, in \emph{Proceedings of the 32nd
  International Symposium on Computational Geometry} (2016), pp. 6:1--6:16

\bibitem{Stojmirovic2009GeometricAO}
A.~Stojmirovic, Y.k. Yu, Journal of Computational Biology \textbf{16}(4), 579
  (2009)

\bibitem{CHEN2004}
L.~Chen, R.~Ng, in \emph{Proceedings of the 30th International Conference on
  Very Large Databases} (2004), pp. 792--803

\bibitem{Chen2005}
L.~Chen, M.T. {\"O}zsu, V.~Oria, in \emph{Proceedings of the 2005 ACM SIGMOD
  International Conference on Management of Data} (2005), pp. 491--502

\bibitem{Marteau2009}
P.F. Marteau, IEEE Transactions on Pattern Analysis and Machine Intelligence
  \textbf{31}(2), 306 (2009)

\bibitem{Sankararaman:2013:MMS:2525314.2525360}
S.~Sankararaman, P.K. Agarwal, T.~M{\o}lhave, J.~Pan, A.P. Boedihardjo, in
  \emph{Proceedings of the 21st ACM SIGSPATIAL International Conference on
  Advances in Geographic Information Systems} (2013), pp. 234--243

\bibitem{Wang2013}
X.~Wang, A.~Mueen, H.~Ding, G.~Trajcevski, P.~Scheuermann, E.~Keogh, Data
  Mining and Knowledge Discovery \textbf{26}(2), 275 (2013)

\bibitem{Wagner:1974:SCP:321796.321811}
R.A. Wagner, M.J. Fischer, Journal of the ACM \textbf{21}(1), 168 (1974)

\bibitem{Ukkonen:1985:AAS:4620.4626}
E.~Ukkonen, Information and Control \textbf{64}(1-3), 100 (1985)

\bibitem{Masek1980AFA}
W.J. Masek, M.S. Paterson, Journal of Computer and System Sciences
  \textbf{20}(1), 18 (1980)

\bibitem{Gold:2018:DTW:3266298.3230734}
O.~Gold, M.~Sharir, ACM Transactions on Algorithms \textbf{14}(4), 50 (2018)

\bibitem{ip-ock-01}
R.~Impagliazzo, R.~Paturi, J. Comp. Sys. Sci. \textbf{62}(2), 367 (2001)

\bibitem{DBLP:journals/corr/Bringmann14}
K.~Bringmann, in \emph{Proceedings of the IEEE 55th Annual Symposium on
  Foundations of Computer Science} (2014), pp. 661--670

\bibitem{7354388}
A.~Abboud, A.~Backurs, V.V. Williams, in \emph{Proceedings of the IEEE 56th
  Annual Symposium on Foundations of Computer Science} (2015), pp. 59--78

\bibitem{DBLP:journals/corr/BringmannK15}
K.~Bringmann, M.~K{\"u}nnemann, in \emph{Proceedings of the IEEE 56th Annual
  Symposium on Foundations of Computer Science} (2015), pp. 79--97

\bibitem{DBLP:journals/corr/BackursI14}
A.~Backurs, P.~Indyk, in \emph{Proceedings of the 47th Annual ACM Symposium on
  Theory of Computing} (2015), pp. 51--58

\bibitem{DBLP:journals/jocg/BringmannM16}
K.~Bringmann, W.~Mulzer, JoCG \textbf{7}(2), 46 (2016)

\bibitem{chan2018improved}
T.M. Chan, Z.~Rahmati, Information Processing Letters pp. 72--74 (2018)

\bibitem{kuszmaul2019dynamic}
W.~Kuszmaul, in \emph{Proceedings of the 46th International Colloquium on
  Automata, Languages and Programming} (2019)

\bibitem{DBLP:journals/corr/abs-1005-4033}
A.~Andoni, R.~Krauthgamer, K.~Onak, in \emph{Proceedings of the IEEE 51st
  Annual Symposium on Foundations of Computer Science} (2010), pp. 377--386

\bibitem{chakraborty2018approximating}
D.~Chakraborty, D.~Das, E.~Goldenberg, M.~Koucky, M.~Saks, in \emph{Proceedings
  of the 2018 IEEE 59th Annual Symposium on Foundations of Computer Science
  (FOCS)} (IEEE, 2018), pp. 979--990

\bibitem{Landau1998IncrementalSC}
G.M. Landau, E.W. Myers, J.P. Schmidt, SIAM Journal on Computing
  \textbf{27}(2), 557 (1998)

\bibitem{an-edntc-20}
A.~Andoni, N.S. Nosatzki, in \emph{Proceedings of the IEEE 61st Annual
  Symposium on Foundations of Computer Science} (2020).
\newblock To appear.

\bibitem{HarPeled2010Chapter1R}
S.~Har-Peled, \emph{Geometric approximation algorithms} (American Mathematical
  Soc., 2011), chap. 11, Random Partition via Shifting

\end{thebibliography}
